%% file: main.tex
\documentclass[envcountsame]{llncs}

\usepackage{enumerate}
\usepackage{wrapfig}
\usepackage{amsopn}
\usepackage{amsfonts}
\usepackage{amssymb}

\usepackage{complexity}
\usepackage{mathtools}
\usepackage{tabularx}
\usepackage{environ}
\usepackage{hyperref}
\usepackage{xcolor,colortbl,tikz}
\usepackage{esvect}

\input{content/macros}
\begin{document}
	
	\title{Liveness in Broadcast Networks}
	
	\author{Peter Chini\and Roland Meyer \and Prakash Saivasan}
	
	\institute{TU Braunschweig \\ \email{ \{p.chini, roland.meyer, p.saivasan\}@tu-bs.de}}
	
	\authorrunning{P. Chini, R. Meyer, and P. Saivasan}
	
	\maketitle
		
	\input{content/abstract}

	\input{content/intro_rev}

	\input{content/bcnetworks}

	\input{content/LVP}

	\input{content/fairLVP}

	\input{content/LTL}

	\bibliographystyle{plain}
	\bibliography{content/cite}
\end{document}

%% file: content/macros.tex
\newcommand{\Domain}{D}
\newcommand{\bcnet}{\mathcal{N}}
\newcommand{\bcprot}{P}
\newcommand{\send}[1]{! #1}
\newcommand{\rec}[1]{? #1}
\newcommand{\Ops}[1]{\mathit{Ops}(#1)}
\newcommand{\IConf}{\mathit{CF}_0}
\newcommand{\Conf}{\mathit{CF}}

\newcommand{\Trans}{\mathit{pr}}

\newcommand{\Post}{\mathit{post}}
\newcommand{\Pre}{\mathit{pre}}
\newcommand{\Enabled}{\mathit{enabled}}
\newcommand{\Postofapp}[2]{\Post_{#1}(#2)}

\newcommand{\Enabledoffapp}[2]{\Enabled_{#1}(#2)}
\newcommand{\Kill}{\mathit{Kill}}
\newcommand{\Gen}{\mathit{Gen}}
\newcommand{\type}{\mathit{idx}}
\DeclareMathOperator{\infop}{Inf}
\newcommand{\Inf}[1]{\infop(#1)}
\DeclareMathOperator{\finop}{Fin}
\newcommand{\Fin}[1]{\finop(#1)}
\DeclareMathOperator{\setop}{Set}
\newcommand{\Set}[1]{\setop(#1)}
\DeclareMathOperator{\Pos}{Pos}
\DeclareMathOperator{\Target}{Target}

\DeclareMathOperator{\inc}{inc}
\DeclareMathOperator{\dec}{dec}
\newcommand{\Del}{\mathit{Del}}
\newcommand{\op}{\mathit{op}}

\newcommand{\prop}{\mathcal{A}}
\newcommand{\Next}{\mathcal{X}}
\newcommand{\Until}{\mathcal{U}}
\newcommand{\bcmap}{\lambda}
\newcommand{\fair}{\mathit{fair}}
\newcommand{\sparse}{\mathit{sparse}}
\DeclareMathOperator{\proj}{proj}
\newcommand{\idx}{\mathit{idx}}
\newcommand{\J}{\mathcal{J}}

\newcommand{\bigO}{\mathcal{O}}
\newcommand{\setcon}[1]{\lbrace #1 \rbrace}
\newcommand{\Setcon}[2]{\lbrace #1 \mid #2 \rbrace}
\newcommand{\Naturals}{\mathbb{N}}
\newcommand{\PSet}[1]{\mathcal{P}(#1)}
\newcommand{\existinf}{\exists^{\infty}}

\newcommand{\abs}[1]{|#1|}

\makeatletter
\newcommand {\problemtitle}[1]{\gdef\@problemtitle{#1}}
\newcommand {\problemshort}[1]{\gdef\@problemshort{#1}}
\newcommand {\probleminput}[1]{\gdef\@probleminput{#1}}
\newcommand {\problemparameter}[1]{\gdef\@problemparameter{#1}}
\newcommand {\problemquestion}[1]{\gdef\@problemquestion{#1}}

\NewEnviron{myproblem}{
	\problemtitle{}
	\problemshort{}
	\probleminput{}
	\problemquestion{}
	
	\BODY
	\par\addvspace{.5\baselineskip}
	\noindent
	\framebox[\textwidth]{
		\begin{tabularx}{\textwidth}{@{\hspace{\parindent}} l X c}
			\multicolumn{2}{@{\hspace{\parindent}}l}{ \normalsize \emph{\@problemtitle} \@problemshort} \\
			\normalsize \textbf{Input:} & \normalsize \@probleminput \\
			\normalsize \textbf{Question:} & \normalsize \@problemquestion
		\end{tabularx}
	}
	\par\addvspace{.5\baselineskip}
}

%% file: content/abstract.tex
\begin{abstract}
	We study liveness and model checking problems for broadcast networks,
	a system model of identical clients communicating via message passing.
	The first problem that we consider is \emph{Liveness Verification}.
	It asks whether there is a computation such that one of the clients visits a final state infinitely often. 
	The complexity of the problem has been open since 2010 when it was shown to be $\P$-hard and solvable in $\EXPSPACE$.
	We close the gap by a polynomial-time algorithm.
	The algorithm relies on a characterization of live computations in terms of paths in a suitable graph, combined with a fixed-point iteration to efficiently check the existence of such paths. 
	The second problem is \emph{Fair Liveness Verification}.
	It asks for a computation where all participating clients visit a final state infinitely often.
	We adjust the algorithm to also solve fair liveness in polynomial time.
	
	Both problems can be instrumented to answer model checking questions for broadcast networks against linear time temporal logic specifications.
	The first problem in this context is \emph{Fair Model Checking}.
	It demands that for all computations of a broadcast network, all participating clients satisfy the specification.
	We solve the problem via the Vardi-Wolper construction and a reduction to \emph{Liveness Verification}.
	The second problem is \emph{Sparse Model Checking}.
	It asks whether each computation has a participating client that satisfies the specification.
	We reduce the problem to \emph{Fair Liveness Verification}.
\end{abstract}

%% file: content/intro_rev.tex
\section{Introduction}
\label{Section:Introduction}

Parameterized systems consist of an arbitrary number of identical clients that communicate via some mechanism like shared memory or message passing~\cite{AK86}.
Parameterized systems appear in various applications. 
In distributed algorithms, a group of clients has to form a consensus \cite{Konnov2017}.
In cache-coherence protocols, coherence has to be guaranteed for data shared among threads~\cite{Delzanno2000}.
Developing parameterized systems is difficult. 
The desired functionality has to be achieved not only for a single system instance but for an arbitrary number of clients that is not known a priori. 
The proposed solutions are generally tricky and sometimes buggy \cite{Lamport1999}, which has lead to substantial interest in parameterized verification~\cite{Bloem2015}, verification algorithms for parameterized systems.

Broadcast networks are a particularly successful model for parameterized verification~\cite{Singh2009,Hague2011,Esparza1999,Esparza2013,Durand-Gasselin2015,Bouyer2016,Delzanno2010,Delzanno2012,Chini2018,Bertrand2015,Bertrand2018}. 
A broadcast network consists of an arbitrary number of identical finite-state automata communicating via passing messages. 
We call these automata clients, because they reflect the interaction of a single client in the parameterized system with its environment. 
When a client sends a message (by taking a send transition), at the same time a number of clients receive the message (by taking a corresponding receive transition). 
A client ready to receive a message may decide to ignore it, and it may be the case that nobody receives the message.

What makes broadcast networks interesting is the surprisingly low complexity of their verification problems.  
Earlier works have concentrated on safety verification. 
In the coverability problem, the question is whether at least one participating client can reach an unsafe state. 
The problem has been shown to be solvable in polynomial time~\cite{Delzanno2012}.
In the synchronization problem, all clients need to visit a final state at the same time.
Although seemingly harder than coverability, it turned out to be solvable in polynomial time as well\cite{Fournier2015}.
Both problems remain in $\P$ if the communication topology is slightly restricted \cite{Bertrand2018}, a strengthening that usually leads to undecidability results~\cite{Delzanno2010,Bertrand2018}.

The focus of our work is on liveness verification and model checking.
Liveness properties formulate good events that should happen during a computation. 
To give an example, one would state that every request has to be followed by a response.
In the setting of broadcast networks, liveness verification was studied in \cite{Delzanno2010}.
The problem generalizes coverability in that at least one client needs to visit a final state infinitely many times. 
The problem was shown to be solvable in $\EXPSPACE$ by a reduction to repeated coverability in Petri Nets~\cite{Esparza1994,Esparza2003}.
The only known lower bound, however, is $\P$-hardness \cite{Delzanno2012}.

Our main contribution is an algorithm that solves the liveness verification problem in polynomial time. 
It closes the aforementioned gap. 
We also address a fair variant of liveness verification where all clients participating infinitely often in a computation have to see a final state infinitely often, a requirement known as compassion~\cite{PS08}. 
We give an instrumentation that compiles away compassion and reduces the problem to finding cycles.
By our results, safety and liveness verification have the same complexity, a phenomenon that has been observed in other models as well~\cite{Hague2011,Esparza2014,Durand-Gasselin2015,HMM18}.

At the heart of our liveness verification algorithm is a fixed-point iteration terminating in polynomial time.
It relies on an efficient representation of computations.
We first characterize live computations in terms of paths in a suitable graph.
Since the graph is of exponential size, we cannot immediately apply a path finding algorithm.
Instead, we show that a path exists if and only if there is a path in some normal form, a result inspired by \cite{Fournier2015}. 
Paths in normal form can then be found efficiently by the fixed-point iteration.

Our results yield efficient algorithms for model checking broadcast networks against linear time temporal logic (LTL) specifications \cite{Pnueli1977}.
Formally, we consider two variants of model checking, like for liveness verification.
The first variant incorporates a notion of fairness.
Given a broadcast network and a specification, it demands that in each computation, all clients satisfy the specification.
The second variant asks for each computation having at least one client satisfying it.
We solve both problems by employing the Vardi-Wolper construction \cite{Vardi1986} and by applying the aforementioned fixed-point iteration for liveness verification.
The results show that the given broadcast network only contributes a polynomial factor to the running time needed for model checking.

The paper at hand is an extension of the conference version that appeared in \cite{Chini0S19}.
It features a new section, namely Section \ref{Section:ModelChecking}, which shows how model checking problems for broadcast networks can be solved by applying the fixed-point iteration for liveness verification.
More precise, the section introduces two new model checking problems and presents algorithms as well as the corresponding time-complexity analyses.

\input{content/related}

%% file: content/related.tex
\subsubsection*{Related Work}
\label{Section:RelatedWork}

We already discussed the related work on safety and liveness verification of broadcast networks. 
Broadcast networks~\cite{Delzanno2010,Singh2009,Esparza1999} were introduced to verify ad hoc networks~\cite{KoenigKozioura2006,SaksenaWiblingJonsson2008}. 
Ad hoc networks are reconfigurable in that the number of clients as well as their communication topology may change during the computation. 
If the transition relation is compatible with the topology, safety verification has been shown to be decidable~\cite{JK08}. 
Related studies do not assume compatibility but restrict the topology~\cite{HMM14}. 
If the dependencies among clients are bounded~\cite{M08}, safety verification is decidable independent of the transition relation~\cite{WiesZuffreyHenzinger2010,Zufferey2013}. 
Verification tools turn these decision procedures into practice~\cite{MS10,DOKO13}. 
D'Osualdo and Ong suggested a typing discipline for the communication topology~\cite{DOO16}. 
In~\cite{Bertrand2018}, decidability and undecidability results for reachability problems were proven for a locally changing topology. 
The case when communication is fixed along a given graph was studied in \cite{Abdulla2013}.
Topologies with bounded diameter were considered in \cite{Delzanno2011}. 
Perfect communication where a sent message is received by all clients was studied in \cite{Esparza1999}. 
Networks with communication failures were considered in \cite{Delzanno2012Failure}.
Probabilistic broadcast networks were studied in \cite{Bertrand2014}.
In \cite{Bertrand2015}, a variant of broadcast networks was considered where the clients follow a local strategy.

Broadcast networks are related to the leader-contributor model. 
It has a fixed leader and an arbitrary number of identical contributors that communicate via a shared memory. 
The model was introduced in \cite{Hague2011}.
The case when the leader and all contributors are finite-state automata was considered in \cite{Esparza2013} and the corresponding reachability problem was proven to be $\NP$-complete.
In \cite{Chini2018}, the authors took a parameterized complexity look at the reachability problem and proved it fixed-parameter tractable.
Liveness verification for this model was studied in \cite{Durand-Gasselin2015}.
The authors show that repeated reachability is $\NP$-complete.
From a viewpoint of parameterized complexity, the problem was considered in \cite{Chini2019}.
Networks with shared memory and randomized scheduler were studied in \cite{Bouyer2016}.
For a survey of parameterized verification we refer to \cite{Bloem2015}.

%% file: content/bcnetworks.tex
\section{Broadcast Networks}
\label{Section:BroadcastNetworks}

We introduce the model of broadcast networks of interest in this paper. 
Our presentation avoids an explicit characterization of the communication topology in terms of graphs.
A \emph{broadcast network} is a concurrent system consisting of an arbitrary but finite number of identical clients that communicate by passing messages to each other. 
Formally, it is a pair $\bcnet = (\Domain,\bcprot)$.
The \emph{domain} $\Domain$ is a finite set of messages that can be used for communication. 
A message $a \in \Domain$ can either be sent, $\send{a}$, or received,~$\rec{a}$.
The set $\Ops{\Domain} = \Setcon{\send{a}, \rec{a}}{a \in \Domain}$ captures the communication operations a client can perform.
For modeling the identical clients, we abstract away the internal behavior and focus on the communication with others via~$\Ops{\Domain}$. 
With this, the clients are given in the form of a finite state automaton $\bcprot = (Q, I, \delta)$, where $Q$ is a finite set of states, $I \subseteq Q$ is a set of initial states, and $\delta \subseteq Q \times \Ops{\Domain} \times Q$ is the transition relation. 
We extend $\delta$ to words in $\Ops{\Domain}^*$ and write $q \xrightarrow{w} q'$ instead of $(q,w,q') \in \delta$. 

During a communication phase in $\bcnet$, one client sends a message that is received by a number of other clients.  
This induces a change of the current state in each client participating in the communication. 
We use \emph{configurations} to display the current states of the clients.   
A configuration is a tuple $c~=~(q_1, \dots, q_k) \in Q^k$, $k \in \Naturals$. 
We use $\Set{c}$ to denote the set of client states occurring in $c$.
To access the components of $c$, we use $c[i] = q_i$. 
As the number of clients in the system is arbitrary but fixed, we define the set of all configurations to be $\Conf = \bigcup_{k \in \Naturals} Q^k$. 
The set of \emph{initial configurations} is given by $\IConf =\bigcup_{k \in \Naturals} I^k$. 
The communication is modeled by a transition relation among configurations.
Let $c' = (q'_1, \dots, q'_k)$ be another configuration with $k$ clients and $a \in \Domain$ a message. 
We have a transition $c \xrightarrow{a}_\bcnet c'$ if the following conditions hold:
(1) there is a sender, an $i \in [1..k]$ such that $q_i \xrightarrow{\send{a}} q'_i$,
(2)~there is a number of receivers, a set $R \subseteq [1..k] \setminus \setcon{i}$ such that $q_j \xrightarrow{\rec{a}} q'_j$ for each $j \in R$, and 
(3) all other clients stay idle, for all $j \notin R \cup \setcon{i}$ we have $q_j = q'_j$. 
We use $\type(c \xrightarrow{a}_\bcnet c') = R \cup \setcon{i}$ to denote the indices of clients that contributed to the transition.
For $i \in \type(c \xrightarrow{a}_\bcnet c')$, we let $\Trans_i(c \xrightarrow{a}_\bcnet c') = c[i] \xrightarrow{\send{a} / \rec{a}} c'[i]$ be the contribution of client $i$.
If $i \notin \type(c \xrightarrow{a}_\bcnet c')$, we set $\Trans_i(c \xrightarrow{a}_\bcnet c') = \epsilon$.

We extend the transition relation to words  $w \in \Domain^*$ and write $c \xrightarrow{w}_\bcnet c'$.
Such a sequence of consecutive transitions is called a \emph{computation} of $\bcnet$. 
Note that all configurations appearing in a computation have the same number of clients.
We write $c \rightarrow^*_\bcnet c'$ if there is a word $w \in \Domain^*$ with $c \xrightarrow{w}_\bcnet c'$.
If $\abs{w} \geq 1$, we also use $c \rightarrow^+_\bcnet c'$.
Where appropriate, we skip $\bcnet$ in the index.
The definition of $\type$ can be extended to computations by combining all clients that contribute to one of the transitions.
We also extend the definition of $\Trans_i$ to computations by appending the individual contributions to transitions.

We are interested in infinite computations, sequences \mbox{$\pi=c_0\rightarrow c_1\rightarrow\ldots$} of infinitely many consecutive transitions.
Such a computation is called \emph{initialized}, if $c_0\in\IConf$. 
We use  $\Inf{\pi} = \Setcon{ i\in\Naturals }{\existinf j : i \in \type(c_j \rightarrow c_{j+1})}$ to denote the set of clients that participate in the computation infinitely often.
Let $F \subseteq Q$ be a set of final states.
Then we let $\Fin{\pi} = \Setcon{ i\in\Naturals }{\existinf j : c_j[i]\in F}$ represent the set of clients that visit final states infinitely often.

%% file: content/LVP.tex
\section{Liveness}
\label{Section:Liveness}

We consider the liveness verification problem for broadcast networks.
Given a broadcast network $\bcnet = (\Domain,\bcprot)$ with $\bcprot = (Q, I, \delta)$ and a set of final states $F \subseteq Q$, the problem asks whether there is an infinite initialized computation $\pi$ in which at least one client visits a state from $F$ infinitely often, $\Fin{\pi} \neq \emptyset$.
\begin{myproblem}
	\problemtitle{Liveness Verification}
	\probleminput{A broadcast network $\bcnet = (\Domain,\bcprot)$ and final states $F \subseteq Q$.}
	\problemquestion{Is there an initialized computation $\pi$ with $\Fin{\pi} \neq \emptyset$?}
\end{myproblem}

\emph{Liveness Verification} was introduced as \emph{Repeated Coverability} in\cite{Delzanno2010}.
Our main contribution is the following theorem.
\begin{theorem}
	\label{Theorem:LVPpcomplete}
	Liveness Verification is $\P$-complete.
\end{theorem}
$\P$-hardness is due to~\cite{Delzanno2012}.
Our contribution is a matching polynomial-time decision procedure. 
Key to our algorithm is the following lemma which relates the existence of an infinite computation to the existence of a finite one. 
\begin{lemma}\label{Lemma:Splitting}
	There is an infinite computation $c_0 \rightarrow c_1\rightarrow\ldots$ that visits states in $F$ infinitely often if and only if there is a finite computation of the form $c_0 \rightarrow^* c \rightarrow^+ c$ with $\Set{c} \cap F \neq \emptyset$.
\end{lemma}
If there is a computation of the form $c_0 \rightarrow^* c \rightarrow^+ c$ with $\Set{c} \cap F \neq \emptyset$, then $c \rightarrow^+ c$ can be iterated infinitely often to obtain an infinite computation visiting $F$ infinitely often. 
In turn, in any infinite sequence from $Q^k$ one can find a repeating configuration (pigeon hole principle). 
This in particular holds for the infinite sequence of configurations containing final states.

Our polynomial-time algorithm for \emph{Liveness Verification} looks for an appropriate reachable configuration $c$ that can be iterated. 
The difficulty is that we have a parameterized system, and therefore the number of configurations is not finite. 
Our approach is to devise a finite graph in which we search for a cycle that mimics the cycle on $c$. 
While the graph yields a decision procedure, it will be of exponential size and a naive search for a cycle requires exponential time.
We show in a second step how to find a cycle in polynomial time.

The graph underlying our algorithm is inspired by the powerset construction for the determinization of finite state automata~\cite{Rabin1959}.
The vertices keep track of sets of states $S$ that a client may be in. 
Different from finite-state automata, however, there is not only one client in a state $s\in S$ but arbitrarily (but finitely) many.
As a consequence, a transition from $s$ to $s'$ may have two effects. 
Some of the clients in $s$ change their state to $s'$ while others stay in $s$. 
In that case, the set of states is updated to $S'=S\cup\setcon{s'}$. 
Alternatively, all clients may change their state to $s'$, in which case we get $S'=(S\setminus\setcon{s})\cup\setcon{s'}$. 

Formally, the graph of interest is $G = (V, \rightarrow_G)$. 
Vertices are tuples of sets of states, $V = \bigcup_{k\leq \abs{Q}}\PSet{Q}^{k}$. 
The parameter $k$ will become clear in a moment.
To define the edges, we need some more notation.
For $S \subseteq Q$ and $a \in \Domain$, let
\begin{align*}
	\Postofapp{\rec{a}}{S}= \Setcon{r' \in Q}{\exists r \in S: r \xrightarrow{\rec{a}} r'}
\end{align*} 
denote the set of successors of $S$ under transitions receiving $a$.
The set of states in $S$ where receives of $a$ are \emph{enabled} is denoted by
\begin{align*}
	\Enabledoffapp{\rec{a}}{S} = \Setcon{r \in S}{\Postofapp{\rec{a}}{\setcon{r}} \neq \emptyset}.
\end{align*}

There is a directed edge $V_1\rightarrow_{G} V_2$ from vertex $V_1=(S_1,\ldots,S_k)$ to vertex $V_2 = (S'_1,\ldots,S'_k)$ if the following three conditions are satisfied:
(1) there is an index $j \in [1..k]$, states $s \in S_j$ and $s' \in S'_j$, and an element $a$ from the domain~$\Domain$ such that $s \xrightarrow{!a} s'$ is a send transition.
(2) For each $i \in [1..k] $ there are sets of states $\Gen_i \subseteq \Postofapp{?a}{S_i}$ and \mbox{$\Kill_i \subseteq \Enabledoffapp{\rec{a}}{S_i}$} 
such that 
\begin{align*}
	S'_i = 
	\left\lbrace
	\begin{aligned}
		&(S_i \setminus \Kill_i) \cup \Gen_i, &\text{ for } i \neq j, \\
		&(U_j \setminus \Kill_j) \cup \Gen_j \cup \setcon{s'}, &\text{ for } i = j
	\end{aligned}
	\right.
\end{align*}
where $U_j$ is either $S_j$ or $S_j \setminus \setcon{s}$.
(3) For each index $i \in [1..k]$ and state $q \in \Kill_i$, the intersection $\Postofapp{\rec{a}}{q} \cap \Gen_i$ is non-empty.

Intuitively, an edge in the graph mimics a transition in the broadcast network without making explicit the configurations.
Condition (1) requires a sender, a component~$j$ capable of sending a message $a$.
Clients receiving this message are represented by (2).
The set $\Gen_i$ consists of those states that are reached by clients performing a corresponding receive transition.
These states are added to~$S_i$.
As mentioned above, states can get killed.
If, during a receive transition, all clients move to the target state, the original state will not be present anymore.
We capture those states in the set $\Kill_i$ and remove them from~$S_i$.
Condition (3) is needed to guarantee that each killed state is replaced by a target state.
Note that for component $j$ we add $s'$ due to the send transition. 
Moreover, we need to distinguish whether state \mbox{$s$ gets killed or not.}

The following lemma relates a cycle in the constructed graph with a cyclic computation of the form $c \rightarrow^+ c$. 
It is crucial for our result.
\begin{lemma}\label{Lemma:Cycle}
	There is a cycle $(\setcon{s_1}, \dots, \setcon{s_m}) \rightarrow^+_G (\setcon{s_1}, \dots, \setcon{s_m})$ in $G$ if and only if there is a configuration $c$ with $\Set{c} = \setcon{s_1,\dots s_m}$ and $c \rightarrow^+ c$.
\end{lemma}
The lemma explains the restriction of the nodes in the graph to $k$-tuples of sets of states, with $k\leq\abs{Q}$. 
We explore the transitions for every possible state in $c$, and there are at most $\abs{Q}$ different states that have to be considered.
We have to keep the sets of states separately to make sure that, for every starting state, the corresponding clients perform a cyclic computation.
\begin{proof}
	We fix some notations that we use throughout the proof.
	Let $c \in Q^n$ be any configuration and $s \in \Set{c}$.
	By $\Pos_c(s) = \Setcon{i \in [1..n]}{c[i] = s}$ we denote the positions of $c$ storing state $s$.
	Given a second configuration $d \in Q^n$, we use the set $\Target_c(s,d) = \Setcon{d[i]}{ i \in \Pos_c(s) }$ to represent those states that occur in $d$ at the positions $\Pos_c(s)$.
	Intuitively, if there is a sequence of transitions from $c$ to $d$, these are the target states of those positions of $c$ storing $s$.
	
	Consider a computation $\pi = c \rightarrow^+ c$ with $\Set{c} = \setcon{s_1, \dots, s_m}$.
	We show that there is a cycle $(\setcon{s_1}, \dots, \setcon{s_m}) \rightarrow^+_G (\setcon{s_1}, \dots, \setcon{s_m})$ in $G$.
	To this end, assume $\pi$ is of the form $\pi = c \rightarrow c_1 \rightarrow \dots \rightarrow c_\ell \rightarrow c$.
	Since $c \rightarrow c_1$ is a transition in the broadcast network, there is an edge
	\begin{align*}
		(\setcon{s_1}, \dots, \setcon{s_m}) \rightarrow_G (\Target_c(s_1,c_1), \dots, \Target_c(s_m,c_1))
	\end{align*} 
	in $G$ where each state $s_i$ gets replaced by the set of target states in $c_1$.
	Applying this argument inductively, we get a path in the graph:
	\begin{align*}
		(\setcon{s_1}, \dots, \setcon{s_m}) &\rightarrow_G (\Target_c(s_1,c_1), \dots, \Target_c(s_m,c_1)) \\
		&\rightarrow_G (\Target_c(s_1,c_2), \dots, \Target_c(s_m,c_2)) \\
		&\rightarrow_G \dots \\
		&\rightarrow_G (\Target_c(s_1,c), \dots, \Target_c(s_m,c)).
	\end{align*} 
	Since $\Target_c(s_i,c) = \setcon{s_i}$,we found the desired cycle.
	
	For the other direction, let a cycle $\sigma = (\setcon{s_1}, \dots, \setcon{s_m}) \rightarrow^+_G (\setcon{s_1}, \dots, \setcon{s_m})$ be given.
	We construct from $\sigma$ a computation $\pi = c \rightarrow^+ c$ in the broadcast network such that $\Set{c} = \setcon{s_1, \dots, s_m}$.
	The difficulty in constructing~$\pi$ is to ensure that at any point in time there are enough clients in appropriate states.
	For instance, if a transition $s \xrightarrow{\send{a}} s'$ occurs, we need to decide on how many clients to move to $s'$.
	Having too few clients in $s'$ may stall the computation at a later point:
	there may be a number of sends required that can only be obtained by transitions from $s'$.
	If there are too few clients in $s'$, we cannot guarantee the sends.
	The solution is to start with \emph{enough} clients in any state.
	With invariants we guarantee that at any point in time, the number of clients in the needed states suffices.
	
	Let cycle $\sigma$ be $V_0 \rightarrow_G V_1 \rightarrow_G \dots \rightarrow_G V_\ell$ with \mbox{$V_0 = V_\ell = (\setcon{s_1}, \dots, \setcon{s_m})$.}
	Further, let $V_j = (S^1_j, \dots, S^m_j)$.
	We will construct the computation $\pi$ over configurations in $Q^n$ where $n = m \cdot \abs{Q}^\ell$.
	The idea is to have $\abs{Q}^\ell$ clients for each of the $m$ components of the vertices $V_i$ occurring in $\sigma$.
	To access the clients belonging to a particular component, we split up configurations in $Q^n$ into \emph{blocks}, intervals $I(i) = [(i-1) \cdot \abs{Q}^\ell + 1 \; .. \; i \cdot \abs{Q}^\ell]$ for each $i \in [1..m]$.
	Let $d \in Q^n$ be arbitrary.
	For $i \in [1..m]$, let $B_d(i) = \Setcon{d[t]}{t \in I(i)}$ be the set of states occurring in the $i$-th block of $d$.
	Moreover, we blockwise collect clients that are currently in a particular state $s \in Q$.
	Let the set \mbox{$\Pos_d(i,s) = \Setcon{t \in I(i)}{d[t] = s}$} be those positions of $d$ in the $i$-th block that store state $s$.
	
	We fix the configuration $c \in Q^n$.
	For each component $i \in [1..m]$, in the $i$-th block it contains $\abs{Q}^\ell$ copies of the state $s_i$.
	Formally, $B_c(i) = \setcon{s_i}$.
	Our goal is to construct the computation $\pi = c_0 \rightarrow^+ c_1 \rightarrow^+ \dots \rightarrow^+ c_\ell$ with \mbox{$c_0 = c_\ell = c$} such that the following two invariants are satisfied.
	(1) For each $j \in [0..\ell]$ and \mbox{$i \in [1..m]$} we have $B_{c_j}(i) \subseteq S^i_j$.
	(2) For any state $s$ in a set $S^i_j$ we have $\abs{\Pos_{c_j}(i,s)} \geq \abs{Q}^{\ell-j}$.
	Intuitively, (1) means that during the computation $\pi$ we visit at most those states that occur in the cycle $\sigma$.
	Invariant (2) guarantees that at each configuration $c_j$ there are \emph{enough} clients available in these states.
	
	We construct $\pi$ inductively.
	The base case is given by configuration $c_0 = c$ which satisfies Invariants (1) and (2) by definition.
	For the induction step, assume $c_j$ is already constructed such that (1) and (2) hold for the configuration.
	Our first goal is to construct a configuration $d$ such that $c_j \rightarrow^+ d$ and $d$ satisfies Invariant (2).
	In a second step we construct a computation $d \rightarrow^* c_{j+1}$.
	
	In the cycle $\sigma$ there is an edge $V_j \rightarrow_G V_{j+1}$.
	From the definition of $\rightarrow_G$ we get a component $t \in [1..m]$, states $s \in S^t_j$ and $s' \in S^t_{j+1}$, and an $a \in \Domain$ such that there is a send transition $s \xrightarrow{\send{a}} s'$.
	Moreover, there are sets \mbox{$\Gen_t \subseteq \Postofapp{\rec{a}}{S^t_j}$} and $\Kill_t \subseteq \Enabledoffapp{\rec{a}}{S^t_j}$ such that the following holds:
	\begin{align*}
		S^t_{j+1} = (U_t \setminus \Kill_t) \cup \Gen_t \cup \setcon{s'}.
	\end{align*} 
	Here, $U_t$ is either $S^t_j$ or $S^t_j \setminus \setcon{s}$. We focus on $t$ and take care of other components later.
	We apply a case distinction for the states in $S_{j+1}^t$.
	
	Let $q$ be a state in $S_{j+1}^t \setminus \setcon{s'}$.
	If $q \in \Gen_t$, there exists a $p \in S^t_j$ such that $p \xrightarrow{\rec{a}} q$.
	We apply this transition to $\abs{Q}^{\ell - (j+1)}$ many clients in the $t$-th block of configuration $c_j$.
	If $q \in U_t \setminus \Kill_t$ and $q$ not in $\Gen_t$, then certainly $q \in U_t \subseteq S^t_j$.
	In this case, we let $\abs{Q}^{\ell-(j+1)}$ many clients of block $t$ stay idle in state $q$.
	For state~$s'$, we apply a sequence of sends.
	More precise, we apply the transition $s \xrightarrow{\send{a}} s'$ to $\abs{Q}^{\ell - (j+1)}$ many clients in block $t$ of $c_j$.
	The first of these sends synchronizes with the previously described receive transitions.
	The other sends do not have any receivers.
	For components different from $t$, we apply the same procedure.
	Since there are only receive transitions, we also let them synchronize with the first send of $a$.
	This leads to a computation $\tau$
	\begin{align*}
		c_j \xrightarrow{a} d^1 \xrightarrow{a} d^2 \xrightarrow{a} \dots \xrightarrow{a} d^{\abs{Q}^{\ell-(j+1)}} = d.
	\end{align*}
	We argue that the computation $\tau$ is \emph{valid}: there are enough clients in $c_j$ such that $\tau$ can be carried out. 
	We again focus on component $t$, the reasoning for the other components is similar.
	Let $p \in \Set{c_j} = S_j^t$.
	Note that the equality is due to Invariants (1) and (2).
	We count the clients of $c_j$ in state $p$ (in block~$t$) that are needed to perform $\tau$.
	We need 
	\begin{align*}
		\abs{Q}^{\ell - (j+1)} \cdot \abs{\Postofapp{\rec{a}}{p} \cup \setcon{p,s'}}
		\leq \abs{Q}^{\ell - (j+1)} \cdot \abs{Q} = \abs{Q}^{\ell-j}
	\end{align*}
	of these clients.
	The set $\Postofapp{\rec{a}}{p} \cup \setcon{p,s'}$ appears as a consequence of the case distinction above:
	there may be transitions mapping $p$ to a state in $\Postofapp{\rec{a}}{p}$, it may happen that clients stay idle in $p$, and in the case $p = s$, we need to add $s'$ for the send transition.
	Since $\abs{\Pos_{c_j}(t,p)} \geq \abs{Q}^{\ell-j}$ by Invariant (2), we get that $\tau$ is a valid computation.
	Moreover, note that configuration $d$ satisfies Invariant (2) for $j+1$:
	for each state $q \in S_{j+1}^t$, the computation $\tau$ \mbox{was constructed such that $\abs{\Pos_d(t,q)} \geq \abs{Q}^{\ell-(j+1)}$.}
	
	To satisfy Invariant (1), we need to erase states that are present in $d$ but not in $S_{j+1}^t$.
	To this end, we reconsider the set $\Kill_t \subseteq \Enabledoffapp{\rec{a}}{S_j^t}$.
	For each state $p \in \Kill_t$, we know by the definition of $\rightarrow_G$ that $\Postofapp{\rec{a}}{p} \cap \Gen_t \neq \emptyset$.
	Hence, there is a $q \in S_{j+1}^t$ such that $p \xrightarrow{\rec{a}} q$.
	We apply this transition to all clients in $d$ currently in state $p$ that were not active in the computation $\tau$.
	In case $U_t = S_j^t \setminus \setcon{s}$, we apply the send $s \xrightarrow{\send{a}} s'$ to all clients that are still in $s$ and were not active in $\tau$.
	Altogether, this leads to a computation $\eta = d \rightarrow^* c_{j+1}$. 
	
	There is a subtlety in the definition of $\eta$.
	There may be no send transition for the receivers to synchronize with since $s$ may not need to be erased.
	In this case, we synchronize the receive transitions of $\eta$ with the last send of $\tau$.
	
	Computation $\eta$ substitutes the states in $\Kill_t$ and state $s$, depending on $U_t$, by states in $S_{j+1}^t$.
	But this means that in the $t$-th block of $c_{j+1}$, there are only states of $S_{j+1}^t$ left.
	Hence, $B_{c_{j+1}}(t) \subseteq S_{j+1}^t$, and Invariant (1) holds.
	
	After the construction of $\pi = c \rightarrow^+ c_\ell$, it is left to argue that $c_\ell = c$.
	But this is due to the fact that (1) holds for $c_\ell$ and $S_{\ell}^t = (\setcon{s_1}, \dots, \setcon{s_m})$.
	\qed
\end{proof}

The graph $G$ is of exponential size.
To obtain a polynomial-time procedure, we cannot just search it for a cycle as required by Lemma~\ref{Lemma:Cycle}.
Instead, we now show that if such a cycle exists, then there is a cycle in a certain normal form.
Hence, it suffices to look for a normal-form cycle.
As we will show, this can be done in polynomial time.
We define the normal form more generally for paths.

A path is in \emph{normal form}, if it takes the shape $V_1\rightarrow_G^* V_m\rightarrow_G^* V_n$ such that the following conditions hold.
In the prefix $V_1\rightarrow_G^* V_m$ the sets of states increase monotonically, $V_i \sqsubseteq V_{i+1}$ for all $i\in [1..m-1]$. 
Here, $\sqsubseteq$ denotes the componentwise inclusion.
In the suffix $V_m\rightarrow_G^* V_n$, the sets of states decrease monotonically, $V_i\sqsupseteq V_{i+1}$ for all $i\in [m..n-1]$. 
The following lemma states that if there is a path in the graph, then there is also a path in normal form.
The intuition is that the variants of the transitions that decrease the sets of states 
can be postponed towards the end of the computation. 
\begin{lemma} \label{Lemma:Normalform}
	There is a path from $V_1$ to $V_2$ in $G$ if and only if there is a path in normal form from $V_1$ to $V_2$.	
\end{lemma}
\begin{proof}
	If $V_1 \rightarrow^*_G V_2$ is a path in normal form, there is nothing to prove.
	For the other direction, let $\sigma = V_1 \rightarrow^*_G V_2$ be an arbitrary path.
	To get a path in normal form, we first simulate the edges of $\sigma$ in such a way that no states are deleted.
	In a second step, we erase the states that should have been deleted.
	We have to respect a particular deletion order to obtain a valid path.
	
	Let $\sigma = U_1 \rightarrow_G U_2 \rightarrow_G \dots \rightarrow_G U_\ell$ with $U_1 = V_1$ and $U_\ell = V_2$.
	We inductively construct an increasing path $\sigma_{\inc} = U'_1 \rightarrow_G \dots \rightarrow_G U'_\ell$ with \mbox{$U'_j \sqsupseteq U_i$} or all $i\leq j$ by mimicking the edges of $\sigma$.
	
	For the base case, we set $U'_1 = U_1$.
	Now assume $\sigma_{\inc}$ has already been constructed up to vertex $U'_j$.
	There is an edge $e = U_j \rightarrow_G U_{j+1}$ in $\sigma$.
	Since \mbox{$U'_j \sqsupseteq U_j$}, we can simulate $e$ on $U'_j$:
	all states needed for execution are present in $U'_j$.
	Moreover, we can mimic $e$ such that no state gets deleted.
	This is achieved by setting the corresponding $\Kill$-sets to be empty.
	Hence, we get an edge $U'_j \rightarrow U'_{j+1}$ with $U'_{j+1} \sqsupseteq U'_j$ (no deletion) and \mbox{$U'_{j+1} \sqsupseteq U_{j+1}$ (simulation of $e$).}
	
	The states in $V'_2 = U'_\ell$ that are not in $V_2$ are those states that were deleted along $\sigma$.
	We construct a decreasing path $\sigma_{\dec} = V'_2 \rightarrow_G^* V_2$, deleting all these states.
	To this end, let $V'_2 = (T_1, \dots, T_m)$ and $V_2 = (S_1, \dots, S_m)$.
	An edge in $\sigma$ deletes sets of states in each component $i \in [1..m]$.
	Hence, to mimic the deletion, we need to consider subsets of $\Del = \bigcup_{i \in [1..m]} (T_i \setminus S_i) \times \setcon{i}$.
	Note that the index $i$ in a tuple $(s,i)$ displays the component the state $s$ is in.
	
	Consider the equivalence relation $\sim$ over $\Del$ defined by $(x,i) \sim (y,t)$ if and only if the last occurrence of $x$ in component $i$ and $y$ in component $t$ in the path $\sigma$ coincide.
	Intuitively, two elements are equivalent if they get deleted at the same time and do not appear again in $\sigma$.
	We introduce an order on the equivalence classes:
	$[(x,i)]_{\sim} < [(y,t)]_{\sim}$ if and only if the last occurrence of $(x,i)$ was before the last occurrence of $(y,t)$.
	Since the order is total, we get a partition of $\Del$ into equivalence classes $P_1, \dots, P_n$ such that $P_j < P_{j+1}$.
	
	We construct $\sigma_{\dec} = K_0 \rightarrow_G \dots \rightarrow_G K_n$ with $K_0 = V'_2$ and $K_n = V_2$ as follows.
	During each edge $K_{j-1} \rightarrow_G K_j$, we delete precisely the elements in $P_j$ and do not add further states.
	Deleting $P_j$ is due to an edge \mbox{$e = U_{k} \rightarrow_G U_{k + 1}$} of $\sigma$.
	We mimic $e$ in such a way that no state gets added and set the corresponding Gen sets to the empty set.
	Since we respect the order~$<$ with the deletions, the simulation of $e$ is possible.
	Suppose, we need a state $s$ in component $t$ to simulate $e$ but the state is not available in component $t$ of $K_{j-1}$.
	Then it was deleted before, $(s,t) \in P_1 \cup \dots \cup P_{j-1}$.
	But this contradicts that $s$ is present in $U_{k}$.
	Hence, all the needed states are available.
	
	Since after the last edge of $\sigma_{\dec}$ we have deleted all elements from $\Del$, we get that $K_n = V_2$.
	This concludes the proof.
	\qed
\end{proof}

Using the normal-form result in Lemma~\ref{Lemma:Normalform}, we now give a polynomial-time algorithm to check whether $(\{ s_1 \}, \dots, \{ s_m\}) \rightarrow^+_G (\{ s_1 \}, \dots, \{ s_m\})$.
The idea is to mimic the monotonically increasing prefix of the path by a suitable post operator, the monotonically decreasing suffix by a suitable pre operator, and intersect the two.
The difficulty in computing an appropriate post operator is to ensure that the receive operations are enabled by sends leading to a state in the intersection, and similar for pre.
The solution is to use a greatest fixed-point computation.
In a first Kleene iteration step, we determine the ordinary $\Post^+$ of $(\{ s_1 \}, \dots, \{ s_m\})$  and intersect it with the $\Pre^*$. 
In the next step, we constrain the $\Post^+$ and the $\Pre^*$ computations to visiting only states in the previous intersection. 
The results are intersected again, which may remove further states.
Hence, the computation is repeated relative to the new intersection.
The thing to note is that we do not work with standard post and pre operators but with operators that are constrained by (tuples of) sets of states.

For the definition of the operators, consider $C = (C_1,\dots,C_m)\in {\PSet{Q}}^m$ for an $m \leq \abs{Q}$.
Given a sequence of sets of states $X_1,\dots,X_m$ where each $X_i \subseteq C_i$, we define $\Post_C(X_1,\dots,X_m)=(X'_1,\dots,X'_m)$ with 
\begin{align*}
	X'_i\; = &\; \Setcon{s'\in Q}{\exists s \in X_i : s \xrightarrow{\send{a}}_{\bcprot \downarrow_{C_i}} s'}\\
	\cup &\; \Setcon{ s'\in Q}{\exists s_1,s_2 \in X_\ell: \exists s \in X_i: s_1 \xrightarrow{\send{a}}_{\bcprot \downarrow_{C_\ell}}  s_2\wedge  s \xrightarrow{\rec{a}}_{\bcprot \downarrow_{C_i}} s'}\; .
\end{align*} 
Here, $\bcprot\!\downarrow_{C_{i}}$ denotes the automaton obtained from $\bcprot$ by restricting it to the states $C_{i}$.
Similarly, we define $\Pre_C(X_1,\dots,X_m)=(X'_1,\dots,X'_m)$ with 
\begin{align*}
	X'_i\; = &\; \Setcon{s\in Q}{\exists s' \in X_i: s \xrightarrow{\send{a}}_{\bcprot \downarrow_{C_i}} s'}\\
	\cup &\; \Setcon{ s\in Q}{\exists s_1,s_2 \in X_\ell: \exists s' \in X_i: s_1 \xrightarrow{\send{a}}_{\bcprot \downarrow_{C_\ell}}  s_2 \wedge s \xrightarrow{\rec{a}}_{\bcprot \downarrow_{C_i}} s'}\ . 
\end{align*}
The next lemma shows that the (reflexive) transitive closures of these operators can be computed in polynomial time. 
\begin{lemma}\label{Lemma:PolyPostPre}
	The closures $\Post^+_C(X_1,\dots,X_m)$ and $\Pre^*_C(X_1,\dots,X_m)$ can be computed in polynomial time.
\end{lemma}
\begin{proof}
	Both closures can be computed by a saturation.
	For $\Post^+_C(X_1, \dots, X_m)$, we keep $m$ sets $R_1, \dots, R_m$, each being the post of a component.
	Initially, we set $R_i = X_i$.
	The defining equation of $X'_i$ in $\Post^+_C(X_1, \dots, X_m)$ gives the saturation.
	One just needs to substitute $X_i$ by $R_i$ and $X_\ell$ by $R_\ell$ on the right hand side.
	The resulting set of states is added to $R_i$.
	This process is applied consecutively to each component and then repeated until the sets $R_i$ do not change anymore, the fixed point is reached.
	
	The saturation terminates in polynomial time.
	After updating $R_i$ in each component, we either already terminated or added at least one new state to a set $R_i$.
	Since there are $m \leq \abs{Q}$ of these sets and each one is a subset of $Q$, we need to update the sets $R_i$ at most $\abs{Q}^2$ many times.
	For a single of these updates, the dominant time factor comes from finding appropriate send and receive transitions.
	This can be achieved in $\mathcal{O}(\abs{\delta}^2)$ time.
	
	Computing the closure $\Pre^*_C(X_1,\dots,X_m)$ is similar.
	One can apply the above saturation and only needs to reverse the transitions in the client.
	\qed
\end{proof}

As argued above, the existence of a cycle reduces to finding a fixed point.
The following lemma shows that it can be computed efficiently.
\begin{lemma}
	\label{Lemma:Equation}
	There is a cycle $(\{ s_1 \}, \dots, \{ s_m\}) \rightarrow^+_G (\{ s_1 \}, \dots, \{ s_m\})$ if and only if there is a non-trivial solution to the equation
	\begin{align*}
		C = \Post^+_C (\{ s_1 \}, \dots, \{ s_m\}) \cap \Pre^*_C (\{ s_1 \}, \dots, \{ s_m\})\; .
	\end{align*}
	 Such a solution can be found in polynomial time.
\end{lemma}
\begin{proof}
	We use a Kleene iteration to compute the greatest fixed point.
	It invokes Lemma~\ref{Lemma:PolyPostPre} as a subroutine.
	Every step of the Kleene iteration reduces the number of states in $C$ by at least one, and initially there are at most $\abs{Q}$ entries with $\abs{Q}$ states each.
	Hence, we \mbox{terminate after quadratically many iterations.}
	
	It is left to prove correctness. Let \mbox{$(\setcon{s_1}, \dots, \setcon{s_m}) \rightarrow^+_G (\setcon{s_1}, \dots, \setcon{s_m})$} be a cycle in $G$.
	By Lemma \ref{Lemma:Normalform} we can assume it to be in normal form.
	Let $(\setcon{s_1}, \dots, \setcon{s_m}) \rightarrow^+_G C$ be the increasing part and \mbox{$C \rightarrow^*_G (\setcon{s_1}, \dots, \setcon{s_m})$} the decreasing part.
	Then, $C$ is a solution to the equation.
	
	For the other direction, let a solution $C$ be given.
	Since $C$ is contained in $\Post^+_C (\setcon{s_1}, \dots, \setcon{s_m})$ we can construct a monotonically increasing path $(\setcon{s_1}, \dots, \setcon{s_m}) \rightarrow^+_G C$.
	Similarly, since \mbox{$C \subseteq \Pre^*_C (\setcon{s_1}, \dots, \setcon{s_m})$}, we get a  decreasing path $C \rightarrow^*_G (\setcon{s_1}, \dots, \setcon{s_m})$.
	Hence, we get the desired cycle.
	\qed
\end{proof}

What is yet open is the question on which states $s_1$ to $s_m$ to perform the search for a cycle. 
After all, we need that the corresponding configuration is reachable from an initial configuration. 
The idea is to use the set of all states reachable from an initial state in the client.
Note that there is a live computation if and only if there is a live computation involving all those states. 
Indeed, if a state is not active during the cycle, the corresponding clients will stop moving after an initial set-up phase.
Since the states reachable from an initial state can be computed in polynomial time \cite{Delzanno2012}, \mbox{the proof of Theorem \ref{Theorem:LVPpcomplete} follows.}

\emph{Liveness Verification} does not take fairness into account. 
A client may contribute to the live computation (and help the distinguished client reach a final state) without ever making progress towards its own final state.

%% file: content/fairLVP.tex
\section{Fair Liveness}
\label{Section:FairLiveness}

We study the fair liveness verification problem that strengthens the requirement on the computation.
Given a broadcast network $\bcnet = (D, \bcprot)$ with clients \mbox{$\bcprot = (Q, I, \delta)$} and a set of final states $F \subseteq Q$, the problem asks whether there is an infinite initialized computation $\pi$ in which clients that send or receive infinitely often also visit their final states infinitely often, $\Inf{\pi}\subseteq \Fin{\pi}$.
Computations satisfying the requirement are called \emph{fair}, a notion that dates back to~\cite{PS08} where it was introduced as compassion or strong fairness.
\begin{myproblem}
	\problemtitle{Fair Liveness Verification}
	\probleminput{A broadcast network $\bcnet = (\Domain,\bcprot)$ and final states $F \subseteq Q$.}
	\problemquestion{Is there an initialized computation $\pi$ with $\Inf{\pi}\subseteq \Fin{\pi}$?}
\end{myproblem}

We solve the problem by applying the cycle finding algorithm from Section~\ref{Section:Liveness} to an instrumentation of the given broadcast network.
Formally, given an instance $(\bcnet, F)$ of \emph{Fair Liveness Verification}, we construct a new broadcast network $\bcnet_F$, containing several copies of $Q$.
Recall that $Q$ is the set of client states in $\bcnet$.
The construction ensures that cycles over $Q$ in $\bcnet_F$ correspond to cycles in $\bcnet$ where each participating client sees a final state.
Such cycles make up a fair computation.
The main result is the following.
\begin{theorem}\label{Theorem:Instrumentation}
	Fair Liveness Verification is $\P$-complete.
\end{theorem}

Hardness follows from \cite{Delzanno2012}.
To explain the aforementioned instrumentation, we need the notion of a good computation, where good means fairness is respected.
A computation \mbox{$c_1 \rightarrow c_2 \rightarrow \dots \rightarrow c_n$} with $n > 1$ is called \emph{good for $F$}, denoted $c_1 \Rightarrow_F c_n$, if every client that makes a move during the computation also sees a final state.
Formally, if \mbox{$i \in \type(c_1 \rightarrow^+ c_n)$} then there exists a $k \in [1..n]$ such that $c_k[i] \in F$. 
The following strengthens Lemma~\ref{Lemma:Splitting}.
\begin{lemma}\label{Lemma:NewSplitting}
	There is a fair computation from $c_0$ if and only if $c_0\rightarrow^*c\Rightarrow_F c$.
\end{lemma}
\begin{proof}
	If there is a computation of the form $c_0 \rightarrow^* c \Rightarrow_F c$, then the good cycle $c \Rightarrow_F c$ can be iterated infinitely often to obtain a fair computation. 	
	
	For the other direction, let a fair computation $\pi$ starting in $c_0$ be given.
	Since the configurations visited by $\pi$ are over $Q^k$ for some $k$, there is a configuration $c$ that repeats infinitely often in $\pi$.
	Hence, we obtain a prefix $c_0 \rightarrow^* c$.
	
	Let $\Inf{\pi} = \setcon{i_1,\dots, i_n}$ be the clients that participate infinitely often in computation $\pi$.
	By the definition of a fair computation, we have \mbox{$\Inf{\pi} \subseteq \Fin{\pi}$.}
	This means that each of the clients in $\setcon{i_1,\dots, i_n}$ visits a state from $F$ infinitely often along $\pi$.
	Hence, for each $j \in [1..n]$ we can find a subcomputation \mbox{$\pi_j = c \rightarrow^+ c$} of $\pi$, in which client $i_j$ visits a state from $F$.
	Combining all $\pi_j$ yields desired good cycle $c \Rightarrow_F c$.
	\qed
\end{proof}

The broadcast network $\bcnet_F$ is designed to detect good cycles $c \Rightarrow_F c$.
The idea is to let the clients compute in phases.
The original state space $Q$ is the first phase.
As soon as a client participates in the computation, it moves to a second phase given by a copy $\hat{Q}$ of $Q$. 
From this copy it enters a third phase $\tilde{Q}$ upon seeing a final state. 
From $\tilde{Q}$ it may return to $Q$. 

Let the given broadcast network be $\bcnet=(D, P)$ with $P=(Q, I, \delta)$.
We define $\bcnet_F = (D\cup\setcon{n}, P_F)$ with fresh symbol $n\notin D$ and extended client
\begin{align*}
	P_F = (\bar{Q}, \tilde I, \bar{\delta})\quad\text{where}\quad
	\bar{Q}=Q\cup \hat{Q} \cup \tilde{Q}.
\end{align*}
Note that the initial phase is $\tilde{Q}$.
For every transition $(q, \op, q') \in \delta$, we have $(q, \op, \hat{q}'), (\hat{q}, \op, \hat{q}'), (\tilde{q}, \op, \tilde{q}')\in \bar{\delta}$. 
For every final state $q \in F$ we have \mbox{$(\hat{q},\send{n},\tilde{q})\in \bar{\delta}$.}
Finally, for every state $q\in Q$ we have $(\tilde{q}, \send{n}, q)\in \bar{\delta}$.
Configuration $c$ admits a good cycle if and only if there is a cycle at $c$ in the instrumented network.
An initial prefix can be mimicked by computations within~$\tilde{Q}$.
\begin{lemma}\label{Lemma:SoundComplete}
	$c_0\rightarrow^*c\Rightarrow_F c$ in $\bcnet$ if and only if 
	$\tilde c_0\rightarrow^*c\rightarrow^+ c$ in $\bcnet_F$. 
\end{lemma}

In the proof, we argue that a good cycle entails a cycle in the instrumented broadcast network.
The reasoning for a prefix is simpler.
For the reverse direction, note that in $c$ all clients are in states from $Q$.
As soon as a client participates in the computation, it will move to $\hat{Q}$.
To return to $Q$, the client will have to see a final state.
This makes the computation good.
\begin{proof}
	First note that by construction, a prefix $c_0 \rightarrow^* c$ exists in $\bcnet$ if and only if $\tilde{c}_0 \rightarrow^* \tilde{c}$ exists in $\bcnet_F$.
	Since one can always move from configuration $\tilde{c}$ to $c$ in $\bcnet_F$, we obtain that $c_0 \rightarrow^* c$ in $\bcnet$ if and only if $\tilde{c}_0 \rightarrow^* c$ in $\bcnet_F$.
	
	It is left to show that there is a good cycle $c \Rightarrow_F c$ if and only if there is a cycle $c \rightarrow^+ c$ in $\bcnet_F$.
	We first show that good cycles entail cycles in $\bcnet_F$.
	To this end, we prove a slightly more general result:
	if there is a computation $\sigma = c \xrightarrow{w}_\bcnet d$, we can derive a computation $c \xrightarrow{w'}_{\bcnet_F} d(\sigma)$ where $w'$ is a word over $\Domain \cup \setcon{n}$ such that its projection to $\Domain$ is $w$.
	The configuration $d(\sigma)$ is defined componentwise.
	Assume $d \in Q^k$ and let $i \in [1..k]$.
	Moreover, let $d[i] = q$, then the $i$-th component of $d(\sigma)$ is defined via a case distinction:
	
	(1) $d(\sigma)[i] = q$ if $\Trans_i(\sigma) = \varepsilon$.
	So if client $i$ does not contribute to $\sigma$, it will not change its state. Note that $c[i] = q$ in this case.
	
	(2) $d(\sigma)[i] = \tilde{q}$ if $\Trans_i(\sigma) = c[i] \rightarrow^+ e[i] \rightarrow^* d[i] = q$ with $e[i] \in F$.
	Here, $e$ is some intermediary configuration visited along $\sigma$.
	Hence, client $i$ visits a final state during the computation, leading it into $\tilde{Q}$ in the instrumented network.
	
	(3) $d(\sigma)[i] = \hat{q}$ otherwise.
	In this case, client $i$ contributes to computation $\sigma$ but does not visit a final state along it.

	We construct the computation $c \xrightarrow{w'}_{\bcnet_F} d(\sigma)$ by induction.
	In the case \mbox{$w = \varepsilon$}, we set $w' = \varepsilon$ and obtain $d(\sigma) = d$.
	Hence, we get the desired (empty) computation in the network $\bcnet_F$.
	
	Now let $w = u . a$ for a word $u \in \Domain^*$ and $a \in \Domain$.
	Then the computation $\sigma = c \xrightarrow{w}_\bcnet d$ splits into two parts.
	There is an intermediary configuration $f$ with $\sigma = \sigma_1 . \sigma_2$ where $\sigma_1 = c \xrightarrow{u}_\bcnet f$ and $\sigma_2 = f \xrightarrow{a}_\bcnet d$.
	We apply the induction to $\sigma_1$ and obtain a computation $\sigma'_1 = c \xrightarrow{u'}_{\bcnet_F} f(\sigma_1)$.
	Let $\J = \type(\sigma_2)$ be the clients contributing to transition $\sigma_2$.
	We extend the computation $\sigma'_1$ for each $j \in \J$ as follows.
	Let $f[j] = q \in Q$ and $d[j] = p \in Q$ and let $\Trans_j(\sigma_2) = q \xrightarrow{\op} p$ with $\op \in \Ops{\Domain}$.
	
	(1) If $\Trans_j(\sigma_1) = \varepsilon$, then by definition $f(\sigma_1)[j] = q$.
	Hence, by construction of $\bcnet_F$, there is a transition $q \xrightarrow{\op} \hat{p}$ which we perform in the $j$-th component.
	
	(2) If $\Trans_j(\sigma_1)$ visits an intermediary final state of $F$, then $f(\sigma_1)[j] = \tilde{q}$.
	In this case, we extend $\sigma'_1$ by the transition $\tilde{q} \xrightarrow{\op} \tilde{p}$ in the $j$-th component.
	
	(3) If $\Trans_j(\sigma_1) \neq \varepsilon$ and does not visit a final state, we get that $f(\sigma_1)[j] = \hat{q}$.
	Then we extend $\sigma_1$ by the transition $\hat{q} \xrightarrow{\op} \hat{p}$ in the $j$-th component.
	
	Clients $j \notin \type(\sigma_2)$ do not change their current state during the extension of $\sigma'_1$.
	Altogether, we obtain a computation of the form $c \xrightarrow{u'}_{\bcnet_F} f(\sigma_1) \xrightarrow{a}_{\bcnet_F} g$, where $g$ is the resulting configuration of the above extension.
	We extend the computation such that it reaches $d(\sigma)$.
	To this end, let $j \in [1..k]$.
	If $g[j] = \hat{p}$ and $p \in F$, we add the transition $\hat{p} \xrightarrow{\send{n}} \tilde{p}$ without any receiving clients in $\bcnet_F$.
	By this, we get that all clients visiting a final state along $\sigma$, finally move to phase $\tilde{Q}$.
	Hence, we obtain the following desired computation:
	\begin{align*}
		c \xrightarrow{u'}_{\bcnet_F} f(\sigma_1) \xrightarrow{a}_{\bcnet_F} g \xrightarrow{n}_{\bcnet_F} \cdots \xrightarrow{n}_{\bcnet_F} d(\sigma).
	\end{align*}
	
	Now assume a good cycle $\pi = c \Rightarrow_F c$ in $\bcnet$ is given.
	We apply the above result and obtain a computation $c \rightarrow^+_{\bcnet_F} c(\pi)$.
	Since $\pi$ is good for $F$, each participating client visits a final state along $\pi$.
	Hence, the configuration $c(\pi)$ can be described as follows.
	Let $j \in [1..k]$ and $c[j] = q$.
	We have
	\begin{align*}
		c(\pi)[j] = \left\lbrace
		\begin{aligned}
			q,& ~ \text{if} ~ \Trans_j(\pi) = \varepsilon, \\
			\tilde{q},& ~ \text{otherwise}.
		\end{aligned}
		\right.
	\end{align*}
	Note that there is not state in phase $\hat{Q}$.
	
	We extend the computation $c \rightarrow^+_{\bcnet_F} c(\pi)$ for each $j \in [1..k]$ with $\Trans_j(\pi) \neq \varepsilon$.
	Since $c(\pi)[j] = \tilde{q}$ in this case, we add the transition $\tilde{q} \xrightarrow{\send{n}} q = c[j]$.
	Hence, we obtain a cycle
	\begin{align*}
		c \rightarrow^+_{\bcnet_F} c(\pi) \xrightarrow{\send{n}}_{\bcnet_F} \cdots \xrightarrow{\send{n}}_{\bcnet_F} c.
	\end{align*}
	
	For the other direction, assume that a cycle $c \rightarrow^+_{\bcnet_F} c$ in $\bcnet_F$ is given.
	By construction, such a cycle can be decomposed as follows:
	\begin{align*}
		c_0 = c \xrightarrow{w_1} d_1 
		\xrightarrow{n}^*c_1 
		\xrightarrow{w_2} d_2 
		\xrightarrow{n}^* c_2 
		\rightarrow \cdots \rightarrow c_{n-1}
		\xrightarrow{w_n} d_n
		\xrightarrow{n}^* c.
	\end{align*}
	Here, each $w_i \in \Domain^*$ and the computation between $d_i$ and $c_i$ is a repeated sending of $n$ without receivers, needed to transition between different phases.
	
	Out of this computation, we construct a good cycle $c \Rightarrow_F c$.
	To this end, we define for each configuration $d$ of $\bcnet_F$ with $k$ components a configuration $\bcnet(d)$ that retrieves the original states out of $d$.
	For each $j \in [1..k]$, we set
	\begin{align*}
		\bcnet(d)[j] = q ~\text{if}~ d[j] = q / \tilde{q} / \hat{q}.
	\end{align*}

	By construction of $\bcnet_F$, we can mimic the transitions occurring in $c_{i-1} \xrightarrow{w_i}_{\bcnet_F} d_i$ by transitions on $Q$ and hence infer a corresponding computation $\bcnet(c_{i-1}) \xrightarrow{w_i}_{\bcnet} \bcnet(d_i)$ on the original network.
	A computation $d_i \xrightarrow{n}^*_{\bcnet_F} c_i$ can be ignored on the original network since $\bcnet(d_i) = \bcnet(c_i)$.
	The reason is that sending $n$ only shifts copies within $\bcnet_F$.
	Altogether, we obtain a cycle
	\begin{align*}
		c = \bcnet(c_0) \xrightarrow{w_1} \bcnet(d_1) \xrightarrow{w_2} \bcnet(d_2) \xrightarrow{w_3} \cdots \xrightarrow{w_n} \bcnet(d_n) = c.
	\end{align*}
	
	The obtained cycle in $\bcnet$ is good for $F$.
	Let $j \in [1..k]$ be a client participating in the cycle.
	Then, $j \in \type(c \rightarrow^+_{\bcnet_F} c)$.
	This means that $j$ starts in a state $q \in Q$ and reaches back to $q$.
	By construction of $\bcnet_F$, this is only possible if $j$ visits a state $p$ from $F$ during the cycle.
	Hence, $p$ is also visited during the obtained cycle on $\bcnet$.
	\qed
\end{proof}

For the proof of Theorem \ref{Theorem:Instrumentation}, it is left to state the algorithm for finding a computation $\tilde c_0\rightarrow^*c \rightarrow^+ c$ in $\bcnet_F$. 
We compute the states reachable from an initial state in $\bcnet_F$. 
As we are interested in a configuration $c$ over $Q$, we intersect this set with $Q$. 
Both steps can be done in polynomial time. 
Let $s_1, \dots, s_m$ be the states in the intersection. 
To these we apply the fixed-point iteration from Lemma~\ref{Lemma:Equation}.
By Lemma \ref{Lemma:Cycle}, the iteration witnesses the existence of a cycle over a configuration $c$ of $\bcnet_F$ that involves only the states $s_1$ up to $s_m$.

%% file: content/LTL.tex
\section{Model Checking Broadcast Networks}
\label{Section:ModelChecking}

We consider model checking for broadcast networks against linear time specifications.
Given a specification, described in linear time temporal logic (LTL)~\cite{Pnueli1977}, we test whether all (infinite) computations of a broadcast network satisfy the specification.
Like for liveness verification, we consider two variants of the problem.
These differ in when a computation of a broadcast network actually satisfies a specification.
\emph{Fair Model Checking} demands that the individual computations of all clients participating infinitely often satisfy the specification.
\emph{Sparse Model Checking} asks for at least one client that participates infinitely often and satisfies the specification.
We show how to solve both problems by reducing them to \emph{Liveness Verification} and \emph{Fair Liveness Verification} via incorporating the Vardi-Wolper construction \cite{Vardi1986}.

Before we consider both mentioned problems, we briefly recall syntax and semantics of linear time temporal logic.
Let $\prop$ be a set of atomic propositions.
An LTL formula $\varphi$ over the set of propositions is defined as follows:
\begin{align*}
	\varphi \; ::= \; p \in \prop  \, | 
	\, \neg \varphi \, | 
	\, \varphi \vee \varphi \, | 
	\, \Next \varphi  \,  | 
	\, \varphi \, \Until \varphi.
\end{align*}

An LTL formula consists of propositions and combinations of the same via negation, union, \emph{next operator} $\Next$, and \emph{until operator} $\Until$.
Semantics is defined in terms of the \emph{satisfaction relation} $\models$.
It describes when a formula is fulfilled.
We define it inductively along the structure of LTL formulas.
Let $w \in \PSet{\prop}^\omega$ be an infinite word consisting of sets of propositions and let $i \in \Naturals$.
We have
\begin{align*}
	w,i &\models p 
	&\text{iff}  \ \
	&p \in w(i), \\
	w,i &\models \neg \varphi 
	&\text{iff} \ \ 
	&w,i \not\models \varphi, \\
	w,i &\models \varphi_1 \vee \varphi_2 
	&\text{iff} \ \ 
	&w,i \models \varphi_1 ~ \text{or} ~ w,i \models \varphi_2, \\	
	w,i &\models \Next \varphi 
	&\text{iff} \ \ 
	&w,i+1 \models \varphi, \\
	w,i &\models \varphi_1 \Until \varphi_2 
	&\text{iff} \ \ 
	&\exists j \geq i : w,j \models \varphi_2 ~ \text{and} ~ \forall k \in [i..j-1]: w,k \models \varphi_1, \\		
	w &\models \varphi 
	&\text{iff} \ \
	&w,0 \models \varphi.
\end{align*}

For model checking broadcast networks against LTL specifications, we need to define when a network computation satisfies an LTL formula.
Since network computations consist of concurrently running client computations, the definition of satisfaction is based upon the latter.

Let $\bcnet = (\Domain,\bcprot)$ be a broadcast network with clients $\bcprot = (Q, I, \delta)$ and let $\varphi$ be an LTL formula.
Moreover, let $\bcmap: Q \rightarrow \PSet{\prop}$ be a map associating to each state of $\bcprot$ a set of atomic proposition.
An infinite computation \mbox{$\sigma = q_0 \rightarrow q_1 \rightarrow \cdots$} of $\bcprot$ with $q_0 \in I$ is said to \emph{satisfy} formula $\varphi$ if $\bcmap(\sigma) \models \varphi$. 
Here, \mbox{$\bcmap(\sigma) = \bcmap(q_0) . \bcmap(q_1) \ldots \in \PSet{\prop}^\omega$} is the infinite word obtained by applying $\bcmap$ to the whole computation.
We also write $\sigma \models \varphi$.

For network computations, we consider two different notions of satisfaction for LTL formulas.
Like for the liveness verification problems, one of the notions incorporates a fairness assumption, the other one focuses on a single client.

Let $\bcnet = (\Domain,\bcprot)$ be a broadcast network with client states $Q$, \mbox{$\bcmap : Q \rightarrow \PSet{\prop}$} a map, and $\varphi$ an LTL formula.
An initialized infinite computation $\pi$ of $\bcnet$ \emph{satisfies} $\varphi$ \emph{under fairness} if each client $i \in \Inf{\pi}$ satisfies $\varphi$ with its contribution, $\Trans_i(\pi) \models \varphi$.
We write $\pi \models_{\fair} \varphi$.
Hence, satisfaction under fairness ensures that each client that participates infinitely often also satisfies the specification.
An initialized infinite computation $\pi$ of $\bcnet$ \emph{sparsely satisfies} $\varphi$ if there is a client $i \in \Inf{\pi}$ that satisfies $\varphi$, $\Trans_i(\pi) \models \varphi$.
We denote it by $\pi \models_{\sparse} \varphi$.
In contrast to fairness, a single client satisfying the specification suffices.

The two notions of satisfaction yield two decision problems: \emph{Fair Model Checking} and \emph{Sparse Model Checking}.
We state both in the subsequent sections and develop algorithms based on the results from Sections \ref{Section:Liveness} and \ref{Section:FairLiveness}.
\subsubsection*{Fair Model Checking}
We want to test whether all computations of a broadcast network satisfy a given LTL formula under the fairness assumption.
Formally, a broadcast network $\bcnet$ is said to \emph{satisfy} an LTL formula $\varphi$ \emph{under fairness}, written $\bcnet \models_{\fair} \varphi$, if for each initialized infinite computation $\pi$ we have that $\pi \models_{\fair} \varphi$.
With this notion at hand, we can formalize the corresponding decision problem \emph{Fair Model Checking}.
Note that we do not explicitly mention the map $\bcmap$ as part of the input. We assume it to be given with the broadcast network.
\begin{myproblem}
	\problemtitle{Fair Model Checking}
	\probleminput{A broadcast network $\bcnet = (\Domain, \bcprot)$ and an LTL formula $\varphi$.}
	\problemquestion{Does $\bcnet \models_{\fair} \varphi$ hold?}
\end{myproblem}

Our goal is to prove the following theorem by presenting an algorithm for \emph{Fair Model Checking}.
Note that the exponential factor in the time estimation only depends on the size of the formula.
The \emph{size} $\abs{\bcnet} = \max\setcon{\abs{\Domain}, \abs{Q}}$ of the broadcast network only contributes a polynomial factor.
\begin{theorem}\label{Theorem:FairModelChecking}
	Fair Model Checking can be solved in time $2^{\bigO({\abs{\varphi}})} \cdot \abs{\bcnet}^{\bigO(1)}$.
\end{theorem}

We need to develop an algorithm for checking $\bcnet \models_{\fair} \varphi$.
A direct iteration over all computations $\pi$ of $\bcnet$ along with a test whether $\pi \models_{\fair} \varphi$ is not tractable since there are infinitely many candidates for $\pi$.
We rather search for a computation that violates $\varphi$.
The non-existence of such a computation ensures that $\bcnet$ satisfies the formula.
Phrased differently, we have $\bcnet \not\models_{\fair} \varphi$ if and only if there is a computation $\pi$ of $\bcnet$ with $\pi \not\models_{\fair} \varphi$.

The latter can be reformulated.
Let $\pi$ be a computation with $\pi \not\models_{\fair} \varphi$.
Then there is a client $i \in \Inf{\pi}$ that violates $\varphi$, we have $\Trans_i(\pi) \models \neg \varphi$.
But by definition this means $\pi \models_{\sparse} \neg \varphi$.
Hence, we obtain that $\pi \not\models_{\fair} \varphi$ if and only if $\pi \models_{\sparse} \neg \varphi$.
The following lemma summarizes the reasoning.
\begin{lemma}\label{Lemma:FairToSparse}
	We have $\bcnet \not\models_{\fair} \varphi$ if and only if there is an initialized infinite computation $\pi$ of $\bcnet$ such that $\pi \models_{\sparse} \neg \varphi$.
\end{lemma}

Due to Lemma \ref{Lemma:FairToSparse} it suffices to find a computation of $\bcnet$ that sparsely satisfies the negation of $\varphi$.
We develop an algorithm for this task along the following two steps.
(1) We employ the Vardi-Wolper construction \cite{Vardi1986} to build a new broadcast network $\bcnet_{\neg \varphi}$ such that all computations of $\bcnet_{\neg \varphi}$ visiting final states infinitely often are computations of $\bcnet$ that sparsely satisfy $\neg \varphi$.
(2) We decide the existence of a former computation with the fixed-point iteration for \emph{Liveness Verification}, developed in Section \ref{Section:Liveness}.

Step (1) relies on the following well-known characterization of LTL formulas in terms of automata.
The result is crucial for our development since it allows for representing all words that satisfy a given formula as a language of a B\"uchi automaton.
We assume familiarity with the notion of B\"uchi automata~\cite{Nerode2001}.
\begin{theorem}\label{Theorem:LTLBuchi}
	\cite{Vardi1986}
	For any LTL formula $\psi$ over $\prop$, one can construct a B\"uchi automaton $B_\psi$ of size at most $2^{\bigO(\abs{\psi})}$ such that for any word $w \in \PSet{\prop}^\omega$
	\begin{align*}
		w \models \psi \ \text{if and only if} \ w \in L(B_\psi).
	\end{align*}
\end{theorem}

Apply Theorem \ref{Theorem:LTLBuchi} to the formula $\neg \varphi$ of interest.
We obtain an automaton $B_{\neg \varphi} = (Q_B, I_B, \delta_B, F_B)$ over the alphabet $\PSet{\prop}$ with states $Q_B$, initial states $I_B$, transitions $\delta_B$, and final states $F_B$.
Further, we assume $B_{\neg \varphi}$ to be non-blocking: in each state we have an outgoing transition on each letter.
The language of $B_{\neg \varphi}$ consists of exactly those words in $\PSet{\prop}^\omega$ that satisfy $\neg \varphi$.

In order to obtain computations of $\bcnet$ that satisfy $\neg \varphi$, we have to build a cross product of $P$, the clients in $\bcnet$, with $B_{\neg \varphi}$.
Intuitively, a computation in the cross product is then a computation of a client which at the same time satisfies the formula $\neg \varphi$.
The construction however needs to take into account that $B_{\neg \varphi}$ is an automaton over a different alphabet than $\Ops{\Domain}$.

Let $P = (Q, I, \delta)$ be the clients of $\bcnet$ and $\bcmap: Q \rightarrow \PSet{\prop}$ the given map.
We define the new client $P_{\neg \varphi} = (Q_{\neg \varphi}, I_{\neg \varphi}, \delta_{\neg \varphi})$ over the alphabet $\Ops{\Domain}$.
The states are given by $Q_{\neg \varphi} = Q \times Q_B$, the initial states by $I_{\neg \varphi} = I \times I_B$.
Transitions in $\delta_{\neg \varphi}$ are defined as follows:
\begin{align*}
	(q,p) \xrightarrow{\op} (q',p') \in \delta_{\neg \varphi}
	\ \ \text{if} \ \
	q \xrightarrow{\op} q' \in \delta
	 \ \ \text{and} \ \
	p \xrightarrow{\bcmap(q)} p' \in \delta_B.
\end{align*}
We also define a set of final states by $F_{\neg \varphi} = Q \times F_B$.
Then, a computation of $P_{\neg \varphi}$ that visits $F$ infinitely often is a computation of $P$ that satisfies $\neg \varphi$.

The broadcast network of interest is $\bcnet_{\neg \varphi} = (\Domain,P_{\neg \varphi})$.
It uses $P_{\neg \varphi}$ as clients and yields the desired result:
a computation of $\bcnet_{\neg \varphi}$ in which a client visits $F_{\neg \varphi}$ infinitely often is a computation of $\bcnet$ that sparsely satisfies $\neg \varphi$.
\begin{lemma}\label{Lemma:FairModelCheckingCorrectness}
	There is a computation $\pi$ of $\bcnet$ such that $\pi \models_{\sparse} \neg \varphi$ if and only if there is a computation $\pi'$ of $\bcnet_{\neg \varphi}$ with $\Fin{\pi'} \neq \emptyset$.
\end{lemma}
Note that the lemma reasons about initialized infinite computations and the set $\Fin{\pi'}$ is defined via the final states $F_{\neg \varphi}$.

Putting Lemmas \ref{Lemma:FairToSparse} and \ref{Lemma:FairModelCheckingCorrectness} together yields a reduction to \emph{Liveness Verification}.
It is left to test whether there is a computation $\pi'$ of $\bcnet_{\neg \varphi}$ in which a client visits $F_{\neg \varphi}$ infinitely often.
Step (2) decides the existence of such a computation by applying the fixed-point iteration from Section \ref{Section:Liveness} to $(\bcnet_{\neg \varphi}, F_{\neg\varphi})$.

Regarding the complexity stated in Theorem \ref{Theorem:FairModelChecking}, consider the following.
By Theorem \ref{Theorem:LTLBuchi}, the automaton $B_{\neg \varphi}$ has at most $2^{\bigO(\abs{\varphi})}$ many states and can be constructed in time $2^{\bigO(\abs{\varphi})}$.
Hence, the clients $P_{\neg \varphi}$ have at most $\abs{Q} \cdot 2^{\bigO(\abs{\varphi})}$ many states.
The size of the alphabet is $\bigO(\abs{\Domain})$.
Constructing the new clients $P_{\neg \varphi}$ can thus be achieved in time polynomial in $2^{\bigO(\abs{\varphi})} \cdot \abs{\bcnet}$.
The size of $\bcnet_{\neg \varphi}$ is at most $2^{\bigO(\abs{\varphi})} \cdot \abs{\bcnet}$.
Since the fixed-point iteration for solving \emph{Liveness verification} runs in polynomial time, we obtain the desired complexity estimation.
It is left to prove Lemma \ref{Lemma:FairModelCheckingCorrectness}. 
\begin{proof}
	First, let $\pi$ be an initialized infinite computation of the broadcast network $\bcnet$ such that \mbox{$\pi \models_{\sparse} \neg \varphi$.}
	Assume that $\pi$ involves $k \in \Naturals$ clients.
	By definition there is a client $i \in \Inf{\pi}$ with $\Trans_i(\pi) \models \neg \varphi$.
	This means that the word $\bcmap(\Trans_i(\pi)) \in \PSet{\prop}^\omega$ satisfies $\neg \varphi$ and thus lies in the language $L(B_{\neg \varphi})$.
	Hence, there is an initialized computation $r_i$ of $B_{\neg \varphi}$ on the word that accepts it.
	Let $b_i \in Q_B^\omega$ be the sequence of states appearing in $r_i$.
	
	Since $B_{\neg \varphi}$ is non-blocking, there is also in initialized computation $r_j$ on the word $\bcmap(\Trans_j(\pi)) \in \PSet{\prop}^\omega$, for each client $j \neq i$.
	Note that $r_j$ is not necessarily accepting.
	We extract the sequence of states $b_j \in Q_B^\omega$ from each $r_j$.
	Denote by $b$ the vector of all sequences 
	\begin{align*}
		b = (b_1, \dots, b_k).
	\end{align*}
	
	We show how to combine the computation $\pi$ and the vector $b$ to obtain a computation $\pi'$ of $\bcnet_{\neg \varphi}$ that satisfies $\Fin{\pi'} \neq \emptyset$.
	To this end, we will use the following notations.
	Let $\ell \in \Naturals$, then the $\ell$-th configuration of $\pi$ is denoted by $\pi(\ell) \in Q^k$.
	For $\ell' \geq \ell$, the notation $\pi(\ell, \ell')$ refers to the subcomputation of $\pi$ starting in $\pi(\ell)$ and ending in $\pi(\ell')$.
	Moreover, we use $b_j(\ell)$ to access the $\ell$-th state in  the sequence $b_j$.
		
	We define $\pi'$ inductively.
	Roughly, each client of $\pi$ gets joined with the corresponding sequence of states in $b$.
	This yields a sequence of configurations of $\bcnet_{\neg \varphi}$.
	Assume we have already constructed the subcomputation $\pi'(0,\ell)$.
	Throughout the induction, we maintain two invariants:
	(1) $\pi'(0,\ell)$ is in fact an initialized computation of $\bcnet_{\neg \varphi}$.
	(2) When we project a configuration visited by $\pi'(0,\ell)$ to $Q$, the client states of $\bcnet$, we obtain the corresponding configuration of~$\pi$.
	Formally, we have $\proj_Q (\pi'(\ell')) = \pi(\ell')$, where the map
	\begin{align*}
		\proj_Q : (Q \times Q_B)^k \rightarrow Q^k
	\end{align*}
	denotes the componentwise projection onto $Q$ and $\ell' \leq \ell$.
	
	We will use a map $\idx_b : [1..k] \rightarrow \Naturals$ to store, for each client $j \in [1..k]$, the latest position up to which $b_j$ has already processed.
	
	As induction basis, we define the initial configuration of $\pi'$:
	\begin{align*}
		\pi'(0) = \pi(0) \times (b_1(0), \dots, b_k(0)).
	\end{align*}
	Here, $\times$ denotes the componentwise product.
	Since $I_{\neg \varphi} = I \times I_B$, Invariant~(1) is satisfied.
	Invariant (2) holds true by construction.
	Moreover, for each client $j \in [1..k]$ we set $\idx_b(j) = 1$.
	
	For the induction step, let $\pi'(0,\ell)$ already be constructed.
	We show how to construct the configuration $\pi'(\ell+1)$ such that $\pi'(0,\ell+1)$ is a proper computation satisfying the invariants (1) and (2).
	Let $\pi'(\ell) = ((q_1,p_1), \dots, (q_k,p_k))$.
	By Invariant (2) we have $\pi(\ell) = (q_1, \dots, q_k)$.
	Moreover, let $\pi(\ell+1) = (q'_1, \dots, q'_k)$.
	
	Since $\pi$ is a computation of $\bcnet$, there is a transition $\tau = \pi(\ell) \xrightarrow{a} \pi(\ell+1)$ on an $a \in \Domain$.
	Let $\J$ denote the set of clients that contribute to $\tau$, \mbox{$\J = \type(\tau)$.}
	We extend the computation $\pi'(0,\ell)$ by a transition $\tau'$ that leads to $\pi'(\ell+1)$.
	
	The configuration $\pi'(\ell+1)$ is defined componentwise as follows:
	\begin{align*}
		\pi'(\ell+1)[j] = \left\lbrace
		\begin{aligned}
			(q'_j, p_j)&, \ \text{if} \ j \notin \J, \\
			(q'_j, p'_j)&, \ \text{otherwise},
		\end{aligned}
		\right.
	\end{align*}
	where $p'_j = b_j(\idx_b(j))$.
	Further we increase $\idx_b(j)$ by $1$ for all $j \in \J$.
	We define the transition $\tau'$ for each client $j \in [1..k]$.
	Consider $\Trans_j(\tau) = q_j \xrightarrow{\op} q'_j$.
	Then we set the contribution of client $j$ to be $(q_j,p_j) \xrightarrow{\op} (q'_j,p'_j)$ with $p'_j$ as defined before.
	The resulting transition is $\tau'$.
	
	Invariant (1) still holds since $\pi'(0,\ell)$ is an initialized computation of $\bcnet_{\neg \varphi}$ and $\tau'$ is a proper transition.
	Note that $\Trans_j(\tau') = (q_j,p_j) \xrightarrow{\op} (q'_j,p'_j)$ is a proper transition in $P_{\neg \varphi}$.
	Hence, $\tau' = \pi'(\ell) \xrightarrow{a} \pi'(\ell+1)$ is a transition in $\bcnet_{\neg \varphi}$.
	Invariant (2) holds by construction of $\pi'(\ell+1)$.
	
	This finishes the construction of $\pi'(0,\ell)$.
	Via induction, we obtain the whole infinite computation $\pi'$ that is a proper computation of $\bcnet_{\neg \varphi}$.
	Since $b_i$ is the sequence of states of $r_i$, an accepting run of $B_{\neg \varphi}$, we get that $\pi'$ visits the set of states $F_{\neg \varphi} = Q \times F_B$ infinitely often.
	Phrased differently, we obtain $\Fin{\pi'} \neq \emptyset$ which finishes the first direction of the proof.
	
	For the other direction, let $\pi'$ be a computation of $\bcnet_{\neg \varphi}$ with \mbox{$\Fin{\pi'} \neq \emptyset$.}
	We assume that $\pi'$ involves $k$ clients.
	Let $i \in \Fin{\pi'}$ be a client of the computation visiting the final states infinitely often.
	
	Consider the computation $\pi = \proj_Q(\pi')$.
	By definition of $\bcnet_{\neg \varphi}$, $\pi$ is an initialized infinite computation of $\bcnet$.
	Since $i \in \Fin{\pi'}$, the word $\bcmap(\Trans_i(\pi))$ lies in the language $L(B_{\neg \varphi})$.
	This means that $\Trans_i(\pi) \models \neg \varphi$.
	Thus, we have a client $i \in \Inf{\pi}$ that satisfies $\neg \varphi$.
	By definition, $\pi \models_{\sparse} \neg \varphi$.
	\qed
\end{proof}
\subsubsection*{Sparse Model Checking}
The second model checking problem that we consider is \emph{Sparse Model Checking}.
It demands that all computations of a broadcast network sparsely satisfy an LTL specification.
We design an algorithm similar to \emph{Fair Model Checking} that invokes the Vardi-Wolper construction to establish a reduction to \emph{Fair Liveness Verification}.
The latter can then be solved \mbox{with the algorithm from Section \ref{Section:FairLiveness}.}

We state the decision problem.
To this end, a broadcast network $\bcnet$ is said to \emph{sparsely satisfy} an LTL formula $\varphi$ if for each initialized infinite computation $\pi$ we have that $\pi \models_{\sparse} \varphi$.
In this case, we write $\bcnet \models_{\sparse} \varphi$.
\begin{myproblem}
	\problemtitle{Sparse Model Checking}
	\probleminput{A broadcast network $\bcnet = (\Domain,\bcprot)$ and an LTL formula $\varphi$.}
	\problemquestion{Does $\bcnet \models_{\sparse} \varphi$ hold?}
\end{myproblem}

The following theorem states the main result.
Note that the size of the broadcast network only contributes a polynomial factor.
\begin{theorem}\label{Theorem:SparseModelChecking}
	Sparse Model Checking can be solved in time $2^{\bigO({\abs{\varphi}})} \cdot \abs{\bcnet}^{\bigO(1)}$.
\end{theorem}

Like for \emph{Fair Liveness Verification} it is simpler to find a computation violating the given specification than checking whether all computations satisfy it.
Hence, we give an algorithm deciding $\bcnet \not\models_{\sparse} \varphi$.
If a computation $\pi$ of $\bcnet$ does not sparsely saitsfy $\varphi$, all clients in $\Inf{\pi}$ will satisfy $\neg \varphi$.
But this means that $\pi \models_{\fair} \neg \varphi$.
We obtain the following lemma.
\begin{lemma}\label{Lemma:SparseToFair}
	We have $\bcnet \not\models_{\sparse} \varphi$ if and only if there is an initialized infinite computation $\pi$ of $\bcnet$ such that $\pi \models_{\fair} \neg \varphi$.
\end{lemma}

As above, we construct a broadcast network $\bcnet_{\neg \varphi}$ with a set of final states~$F_{\neg \varphi}$.
We invoke Theorem \ref{Theorem:LTLBuchi} on $\neg \varphi$ and obtain an automaton $B_{\neg \varphi}$.
Now we construct the clients $P_{\neg \varphi}$ as before, as a cross product of $P$ and $B_{\neg \varphi}$.
The resulting broadcast network is $\bcnet_{\neg \varphi} = (\Domain, \bcprot_{\neg \varphi})$.
The final states are $F_{\neg \varphi} = Q \times F_B$.
Computations of $\bcnet_{\neg \varphi}$ in which each client that moves infinitely often also visits $F_{\neg \varphi}$ infinitely often are precisely those computations of $\bcnet$ that satisfy $\neg \varphi$ under fairness.
The following lemma summarizes the statement.
\begin{lemma}\label{Lemma:SparseModelCheckingCorrectness}
	There is a computation $\pi$ of $\bcnet$ such that $\pi \models_{\fair} \neg \varphi$ if and only if there is a computation $\pi'$ of $\bcnet_{\neg \varphi}$ with $\Inf{\pi'} \subseteq \Fin{\pi'}$.
\end{lemma}

From Lemmas \ref{Lemma:SparseToFair} and \ref{Lemma:SparseModelCheckingCorrectness}, we obtain a reduction to \emph{Fair Liveness Verification}.
Hence, we can apply the algorithm from Section \ref{Section:FairLiveness} to decide the existence of $\pi'$.
Since the algorithm takes only polynomial time, the complexity estimation given in Theorem \ref{Theorem:SparseModelChecking} follows similar to the estimation of Theorem \ref{Theorem:FairModelChecking}.